\documentclass{eptcs}
\usepackage{amsfonts,amsmath,amsthm}

\usepackage{amsmath,amssymb}  % Better maths support & more symbols
\newcommand*{\scrH}{\mathord{\mathcal{H}}}%
\newcommand*{\scrC}{\mathord{\mathcal{C}}}%
\newcommand*{\scrB}{\mathord{\mathcal{B}}}%
\newcommand*{\AAA}{{\mathbb{A}}}% quantum real numbers
\newcommand*{\NN}{{\mathbb{N}}}% natural numbers
\newcommand*{\ZZ}{{\mathbb{Z}}}% integers
\newcommand*{\QQ}{{\mathbb{Q}}}% rational numbers
\newcommand*{\RR}{{\mathbb{R}}}% real numbers
\newcommand*{\scrE}{\mathord{\mathcal{E}}}%
\newcommand*{\scrT}{\mathord{\mathcal{T}}}% 
\newcommand*{\scrK}{\mathord{\mathcal{K}}}%
\newcommand*{\scrN}{\mathord{\mathcal{N}}}%
\newcommand*{\scrP}{\mathord{\mathcal{P}}}% projection
\newcommand*{\scrS}{\mathord{\mathcal{S}}}% Schwartz space
\newcommand*{\Shv}{\mathop{\mathrm{S\lowercase{hv}}}}% Category of Sheaves

\DeclareRobustCommand\openone{\leavevmode\hbox{\small1\normalsize\kern-
.33em1}}%

\newcommand*{\EsubS}{\ensuremath{\mathord{\mathcal{E}_{\scrS}}}}
\newcommand*{\RsubC}{\ensuremath{\mathord{\RR_{\mathrm{C}}}}}
\newcommand*{\RsubD}{\ensuremath{\mathord{\RR_{\mathrm{D}}}}}
\newcommand*{\ZsubD}{\ensuremath{\mathord{\ZZ_{\mathrm{D}}}}}
\newcommand*{\QsubD}{\ensuremath{\mathord{\QQ_{\mathrm{D}}}}}

\newtheorem{theorem}{Theorem}
\newtheorem{definition}{Definition}
\newtheorem{lemma}[definition]{Lemma}
\newtheorem{proposition}[definition]{Proposition}

\newtheorem{corollary}[definition]{Corollary}

\begin{document}
% You should use BibTeX and apsrev.bst for references
\bibliographystyle{eptcs}
%\preprint{HEP/123-qed}

\def\titlerunning{Quantum Foundations}
\def\authorrunning{J.V.Corbett}

\title{A Topos Theory Foundation for Quantum Mechanics}
\author{John V Corbett\\
\email{john.corbett@mq.edu.au}
\institute{Department of Mathematics, Macquarie University, N.S.W. 2109 and\\
The Centre for Time, Philosophy Department, Sydney University, N.S.W. 2006,
Australia }}
%\thanks{also at Physics Department, XYZ University.}
%\author}%
%\email{jvc@ics.mq.edu.au}
%\affiliation{Department of Mathematics,Macquarie University, Sydney,
%N.S.W. 2109,  Australia}%

\date{\today}
% It is always today, today, but you may specify any date with \date.

%Valid PACS numbers may be entered using the \verb+\pacs{#1}+ command.
%\end{abstract}

%\pacs{03.65, 03.67}
% PACS, the Physics and Astronomy Classification Scheme.

\maketitle
\begin{abstract} The theory of quantum mechanics is examined using non-standard
real numbers, called quantum real numbers (qr-numbers), that are constructed
from standard Hilbert space entities. Our goal is to resolve some of the
paradoxical features of the standard theory by providing the physical attributes
of quantum systems with numerical values that are Dedekind real numbers in the
topos of sheaves on the state space of the quantum system. The measured standard
real number values of a physical attribute are then obtained as constant
qr-number approximations to variable qr-numbers.  Considered as attributes, the
spatial locations of massive quantum particles form non-classical spatial
continua in which a single particle can have a quantum trajectory which passes
through two classically separated slits and the two particles in the Bohm-Bell
experiment stay close to each other in quantum space so that Einstein locality
is retained. \end{abstract}

%\tableofcontents
\section{The quantum real number interpretation of quantum mechanics.}
The quantum real number (qr-number) interpretation of quantum mechanics is based on the claim that the change from classical physics to quantum physics is achieved by changing in the type of numerical values that the physical variables can possess. It assumes that the properties of microscopic entities differ from those of classical because their numerical values are not standard real numbers but are Dedekind real numbers in a spatial topos built upon the quantum state space of the microscopic entities. The Dedekind reals differ from standard real is that each holds true to a non-trivial extent given by an open subsets of the quantum state space. These are its conditions, they replace the individual states of standard quantum theory. The internal  logic is intuitionistic. The ``direct connection between observation properties and properties possessed by the independently existing object'' \cite{cush} is cut, an indirect connection is made through the experimental measurement processes. The interpretation builds on the fact that any experimental measurement has a limited level of accuracy. Within this level of accuracy an attribute's qr-number value is approximated by a standard real number.   

Its mathematical structure is built from elements of standard quantum mechanics: we start from von Neumann's assumption \cite{vonneumann} that to any quantum system we can associate a Hilbert space $\mathcal{H}$. $\mathcal{H}$ is the carrier space of a unitary representation of a Lie group $G$, the symmetry group of the quantum system.  Then, as in the standard interpretation, the physical attributes (measurable qualities) of the quantum system are represented by essentially self-adjoint operators that act on a dense subset  $\mathcal{D} \subset \mathcal{H}$. The set of operators form a $\ast$-algebra $\mathcal A$ and the state space $\EsubS(\mathcal{A}) $ of the system is the space of normalized linear functionals  on $\mathcal A$. The state space has the weak topology generated by the real - valued functions $a_Q  :  \EsubS(\mathcal{A}) \to \RR$ given by $ a_{Q}(\rho)  = Tr( \rho.\hat A)\; :\forall \rho\ \in \EsubS(\mathcal{A})$ and labeled by the operators $\hat A \in \mathcal{A}$. A typical open set is a finite intersection of open sets $\mathcal{N}(\rho_{0}; \hat A;  \epsilon) = \{\rho ;\; |Tr \rho \hat A - Tr \rho_{0}\hat A| <\epsilon\}$, for $\hat A \in \mathcal{A},\; \rho_{0} \in \EsubS(A),\; \epsilon > 0$.  The basic open sets are $\nu(\rho_0 ; \delta) = \{ \rho | Tr |\rho_{0} - \rho| < \delta\}$, when $\mathcal{A}$ is a representation  $dU(\mathcal{E}(\mathcal{G}))$ of the enveloping algebra of a Lie group $G$. 

We have argued elsewhere that many of the conceptual enigmas of quantum theory, including the measurement problem \cite{durt2}, the double slit \cite{JVC2} and the non-locality of entangled systems \cite{JVC1}, result from trying to fit the microscopic phenomena into a conceptual framework in which standard real numbers are taken to be the only possible quantitative values for the attributes of quantum systems. We claim that it is as unreasonable to assume that the only numerical values of physical quantities are classical real numbers as it is to assume that the only geometry of space, or space-time, is Euclidean. The qr-number interpretation assumes that any physical quality, represented by an operator $\hat A \in \mathcal{A}$, always has a  qr-number value, in the form of a real - valued function $a_Q(U)\ = a_{Q}|_{U}$, where $U$ is an open subset of $\EsubS(\mathcal{A})$. $U$ is the ontological condition of the system which may be a subset of its experimentally prepared epistemological condition. The
 qr-numbers answers to the conceptual enigmas listed above will be give after the mathematical model has been presented more fully.

\subsection{ Topos theoretical foundations}
In the prologue to \cite{maclane}, Mac Lane and Moerdijk note that ``a topos can be considered both as a 'generalized space' and as a 'generalized universe of sets'.'' Both generalisations are important for the foundations of quantum mechanics. 

The logic of topos theory is intuitionistic which differs from Boolean logic by the absence of LEM, the law of excluded middle, i.e. $a \lor \lnot a = T$ does not hold, and some of the ``paradoxes'' of standard quantum mechanics appear less paradoxical if LEM is excluded. For example, it can be argued that the double slit and Einstein-Podolsky-Rosen paradoxes are due in part from the use of LEM. One of N D Mermin's possible conclusions from his analysis \cite{mermin} of the Einstein-Podolsky-Rosen experiment is to ``Beware of reasoning from what might have happened but didn't''.  In the qr-number analysis of the double slit experiment \cite{JVC2} a quantum particle cannot pass through both slits, $a \land b = F$ so that $b \leq \lnot a$, but  $a \lor b \neq T$ because the particle has qr-number trajectories that cannot be observed passing through either slit, $a$ or $b$ but are observed to have passed through the double slit. The propositions hold to extents, the conditions of the particle, that are  open subsets of $\EsubS(\mathcal{A})$ so that the propositional calculus is a Heyting algebra. 

The topos of sheaves on a topological space has a ring of ``real numbers'' \cite{maclane}.  Isham and Butterfield (2000) claim that most of the ``paradoxes'' of quantum mechanics depend upon ``the discrepancy between - on the one hand -  the values of physical quantities, and - on the other hand  - the results 
of measurements''.  This discrepancy motivates our use of the object of Dedekind reals $\RsubD(\EsubS(\mathcal{A}))$ in the topos $\Shv(\EsubS(\mathcal{A}))$ as the ontic values of physical quantities.    

A distinction between the qualities of a physical system and their quantitative values is assumed. Qualities are the physical attributes of the system and their quantitative values are the numerical values that they take, or have. This distinction was not made in classical physics because the qualities were represented by real numbers which were also their values.  In this interpretation, qualities are represented by operators and their quantitative values are given by qr-numbers \cite{adelman2,durt2}. The algebraic relations between the operators describe relations between the qualities that hold independently of the numerical values that they take; they serve to identify the qualities. The quantitative values of the qualities are given by qr-numbers that depend upon the condition of the system.

The qr-number approach is not the only topos theoretical approach. Topos theory also offers a broad conceptual framework to formalize the notion of {\sl  contextuality}. Isham and Butterfield (2000) \cite{isham} argue that an appeal to an underlying contextual
theory leads to the construction of an appropriate topos, a topos of presheaves on a suitable base category, which represents the different ``viewpoints'' or ``contexts'' in which propositions in quantum mechanics acquire a meaning. Their proposal  doesn't change the real numbers, it changes the truth value of typical propositions of the type ``the value of the quantity $\hat A$ lies in the Borel set $\Delta \subset \RR$'', to the truth values of a set of weaker propositions involving real valued functions defined on the spectrum of the bounded operator $\hat A$.  

It is instructive to compare the qr-number topos theory of quantum mechanics with those of Doering and Isham (2007) \cite{isham2} and of Heunen and Spitters (2007) \cite{heunen}. All start from algebraic formulations of quantum theory;  Doering and Isham use von Neumann  algebras, Heunen and Spitters use $C^{\ast}$-algebras, and we use $O^{\ast}$-algebras. The first two take a non-commutative algebra and cover it by a family of commutative sub-algebras that are objects in a topos of presheaves, geometrically described in \cite{vickers} using spectral bundles. We call this the commutative sub-algebra approach. 

Although both the commutative sub-algebra and the qr-numbers approaches start from Bohr's insight that quantum phenomena are only empirically accessible through classical physics, they have very different understandings of what can count as empirical facts. If physical quantities can be measured with infinite accuracy then only commutative sub-algebras can be measured because Heisenberg's uncertainty principle implies that non-commuting elements can't simultaneously have infinitely accurate real number values. This is the understanding of empirical facts that seems to be accepted in the commutative sub-algebras approach. The qr-numbers approach accepts that empirical numerical facts are only ever known or ever knowable with finite accuracy; as Brouwer had observed, infinite accuracy requires the verification of an infinite number of decimal places. In this view, Heisenberg's uncertainty relations puts limits on the level of accuracy attainable in measuring non-commuting elements but does not prohibit their measurement. Furthermore there is a positivist philosophy flavour to the commutative sub-algebra approach which implies that physical qualities only have numerical values when they are observed. In the qr-number approach the physical qualities have numerical values even when not being observed.

\subsection{The mathematical background}
The mathematical structures underpinning the qr-number interpretation of quantum mechanics are $\ast$-algebras of unbounded operators called $O^{\ast}$-algebras, $\mathcal A$, and their state spaces $\EsubS(\mathcal A)$. The qr-numbers form an object, the Dedekind real numbers, in the spatial topos, $\Shv (\EsubS(\mathcal A))$, of sheaves on the topological space $\EsubS(\mathcal A)$. 

These structures are presented using standard concepts from the theory of operator algebras on Hilbert spaces and point set topology. The discussion of $O^{\ast}$-algebras and their state spaces relies heavily on the books of Schm\"udgen \cite{xxx} and Inoue \cite{inoue}, although it contains some new results specific to the needs of qr-numbers. The discussion of topos theory in the book of Mac Lane and Moerdijk \cite{maclane} is required to fill out the general notion of Dedekind real numbers. The different constructions of real numbers in a topos $\Shv (X)$ of sheaves on a topological space $X$ are reviewed in Stout \cite{stout}. 
The present work contains some results specific to $\Shv (\EsubS(\mathcal A))$ which depend upon its convexity.

Topos theory is a generalisation of set theory that can be used as a framework for mathematics. It allows us to exploit a form of complementarity between logic and structure.  We will use a spatial topos $\Shv(X)$ of sheaves on a topological space $X$. Brouwer's suggestion that we should understand a space through its open sets rather than by its elements was developed further through the concept of a sheaf. Given $X$, a sheaf assigns to each open set $U \in \mathcal O(X)$ some data, $F(U)$, in the form of sheets stacked over $X$ like Riemann surfaces. The data can be restricted to subsets $V \subset U$ or can be collated over $W = \cup U_j$ provided that the data agree on the overlaps $ U_j \cap U_k$.

The logic of topos theory is \emph{intuitionistic}. It retains the law of non-contradiction: that for any proposition $P$, $P \land \lnot P$ is false and excludes the law of the excluded middle: that  $P \lor \lnot P$ is true. In particular, it excludes the law of double negation: that $\lnot \lnot P \Rightarrow P$. Nevertheless we still have the contrapositive rule: that $ P \Rightarrow Q$ is equivalent to $\lnot Q \Rightarrow \lnot P$. The extent to which a proposition is true is given by an open subset of the topological space $X$. The sheaf $\Omega(W) = \{ U | U \subset W, U \in \mathcal{O}(X)\}$, defined for any open subset $W \subseteq X$, is a subobject classifier for $\Shv(X)$.

The $O^{\ast}$-algebras in this paper come from unitary representations $\hat U$ of Lie groups $G$ on a Hilbert space $\mathcal{H}$. The underlying $\ast$-algebra is the enveloping algebra $\mathcal{E}(\mathcal{G})$ of the Lie algebra $\mathcal{G}$ of $G$. In the theory of qr-numbers, the
 $O^{\ast}$-algebra $\mathcal{A}$ is the infinitesimal representation $dU$ of the enveloping algebra $\mathcal{E}(\mathcal{G})$ obtained from a unitary representation $U$ of $G$ on a separable Hilbert space $\mathcal{H}$. The properties of $dU(\mathcal{E}(\mathcal{G}))$, the infinitesimal representation of the enveloping algebra $\mathcal{E}(\mathcal{G})$ obtained from a unitary representation $U$ of the group, are used to determine the topological properties of its state space $\EsubS(dU(\mathcal{E}(\mathcal{G})))$. 

For any topological space $X$, the topos $Shv(X)$ has an object called the {\em Dedekind} reals 
$\RsubD(X)$ which is equivalent to the sheaf of germs of continuous real-valued functions on $X$, $\RsubD(X) \equiv \mathcal C(X)$, two functions having the same germ at $x \in X$ if they agree on some open neighbourhood of $x$. The local sections of $C(X)$ over an open subset $W$ of $X$ are real numbers defined to extent $W$. These real numbers are as mathematically acceptable as the classical real numbers. Which real numbers are used in a physical theory depends upon their fitness for the task of representing the numerical values of the physical qualities of the theory. In the qr-number interpretations \cite{adelman2,arxives,arxives2,durt2}, the topological space $X =  \EsubS(dU(\mathcal{E}(\mathcal{G})))$. 
\subsection{$O^{\ast}$-algebras}
The mathematics of $O^{\ast}$-algebras was developed \cite{powers} from concepts needed for the study of a single unbounded self-adjoint operator and $C^{\ast}$-algebras guided by the theory of representations of Lie algebras. Schm\"udgen \cite{xxx} is the basic reference.
Energy, momentum and position operators have explicit realisations in $O^{\ast}$-algebras.
   
If $\mathcal D$ is a dense subset of a
Hilbert space $\scrH$, let $\mathcal L(\mathcal D)$ (resp. 
$\mathcal L_{c}(\mathcal D)$) denote the 
set of all (resp. all closable) linear operators from $\mathcal D$ to $\mathcal D$. 
If $\hat A^{\ast}$ is the Hilbert space adjoint of a linear operator $\hat A$ whose domain, $\text{dom}(\hat A)$, is $\mathcal{D}$ put
\begin{equation} 
\mathcal{L}^{\dagger}(\mathcal D) = \{ \hat A \in \mathcal L(\mathcal D) ; \text{dom}(\hat A^{\ast})
\supset \mathcal D,  \hat A^{\ast} \mathcal D \subset \mathcal D\}
\end{equation}
Then $\mathcal L^{\dagger}(\mathcal D) \subset \mathcal L_{c}(\mathcal D) \subset \mathcal L(\mathcal D) $ where $ \mathcal L(\mathcal D) $ is an algebra with the usual operations for linear operators with a common invariant domain: addition $\hat A +\hat  B$, scalar multiplication $\alpha \hat A$ and
non-commutative multiplication $\hat A\hat B$. Furthermore $\mathcal L^{\dagger}(\mathcal D)$  is a $\ast$-algebra with the involution $\hat A \to \hat A^{\dagger} :\equiv \hat A^{\ast} |_{\mathcal D}$.
\begin{definition}
An O-algebra on $\mathcal D \subset \scrH$ is a subalgebra of $\mathcal L(\mathcal D)$ that is contained in $\mathcal L_{c}(\mathcal D)$ whilst an $O^{\ast}$-algebra on $\mathcal D \subset \scrH$ is a $\ast$-subalgebra of $\mathcal L^{\dagger}(\mathcal D)$. 
\end{definition} 
Let $\mathcal A$ be an O-algebra on $\mathcal D \subset \scrH$, then the natural graph topology of $\mathcal A$ on $\mathcal D$, denoted $t_{\mathcal A}$, is the locally convex topology defined by the family of seminorms $\{ \| \cdot \|_{\hat A} ; \hat A \in \mathcal A \}$ where $\| \phi \|_{\hat A} =  \| \hat A\phi \|, \phi \in \mathcal D$. It is the weakest locally convex topology on $\mathcal D$ in which each operator $\hat A \in \mathcal A$
is a continuous mapping of $\mathcal D$ into itself. Every $O^{\ast}$-algebra $\mathcal A$  is a directed family (Schm\"udgen \cite{xxx}, Definition 2.2.4) so that a mapping $\hat A$ of $\mathcal D$  into $\mathcal D$ is continuous iff for each semi-norm $\|\cdot\|_{\hat X}$ on  $\mathcal D$ there is a semi-norm $\|\cdot\|_{\hat Y}$ on  $\mathcal D$ and a positive real number $\kappa$  such that 
$\|\hat A \phi  \|_{\hat X}\leq \kappa \|\ \phi \|_{\hat Y}$
for all $\phi \in \mathcal D$. Therefore any $\hat A \in \mathcal A$ defines a continuous map from $\mathcal D$ to $\mathcal D$ by taking $\hat Y = \hat X \hat A$ and $\kappa = 1$ for all $\hat X \in \mathcal M$. This shows that any $\hat A \in \mathcal A$ is bounded as a linear map from $\mathcal D$ to $\mathcal D$.

We will use the pre-compact topology on $\mathcal D$ which is determined by a directed family of semi-norms $p_{\scrK,\scrN } (\hat A) := \sup_{\phi \in \scrK}  \sup_{\psi \in \scrN} | \langle \hat A \phi, \psi\rangle | $ where  $\scrK, \scrN$ range over pre-compact subsets of $\mathcal D$ in the topology, $t_{\mathcal A}$, see Schm\"udgen \cite{xxx}, Section 5.3. If $\mathcal D$  is a complete locally convex Hausdorff space with respect to the topology  $t_{\mathcal A}$, then any subset whose closure is compact  is pre-compact. $\mathcal D$ is a Fr\'echet space if it is metrizable, e.g., when the topology $t_{\mathcal A}$ is generated by a countable number of semi-norms.  

\subsubsection{The $O^{\ast}$-algebra $dU(\mathcal{E}(\mathcal{G}))$.}  
In this paper each $O^{\ast}$-algebra comes from a unitary representation $\hat U$ of a Lie group $G$ on a Hilbert space $\mathcal{H}$. Its $\ast$-algebra is the enveloping algebra $\mathcal{E}(\mathcal{G})$ of the Lie algebra $\mathcal{G}$ of $G$.
Given a unitary representation $U$ of $G$ on $\mathcal{H}$, a vector $\phi \in \mathcal{H}$ is a 
$C^{\infty}$-vector for $U$ if the map $g \to U(g)\phi$ of the $C^{\infty} $ manifold $G$ into $\mathcal{H}$ is a $C^{\infty}$-map. The set of $C^{\infty}$-vectors for $U$ is denoted $\mathcal{D}^{\infty}(U)$, it is a dense linear subspace of $\mathcal H$ which is invariant under $\hat U(g), \; g \in G$. 

The representation of the Lie algebra $\mathcal{G}$ of $G$ is obtained from the unitary representation $U$ of $G$ by defining $\forall x \in \mathcal{G}$ the operator $dU(x)$ with domain $\mathcal{D}^{\infty}(U)$ as
\begin{equation}
dU(x)\phi = \frac{d}{dt}U(\exp tx)\phi|_{t=0} = \lim_{t\to0} t^{-1}(U(\exp tx ) - I)\phi, \; \phi \in \mathcal{D}^{\infty}(U)
\end{equation}

$dU(x)$ belongs to a $\ast$-representation of $\mathcal{G}$ on $\mathcal{D}^{\infty}(U)$ which has a unique extension to a $\ast$-re\-pre\-sen\-ta\-tion of $\mathcal{E}(\mathcal{G})$ on $\mathcal{D}^{\infty}(U)$ called the \emph{infinitesimal representation} of the unitary representation $U$ of $G$. Each operator $\imath  dU(x)$ is essentially self-adjoint on $\mathcal{D}^{\infty}(U)$.
The graph topology $\tau_{dU}$  is generated by the  family of semi-norms  $\| \cdot \|_{dU(y)} $, for $y \in \mathcal{E}(\mathcal{G})$. Then $\mathcal{D}^{\infty}(U) $ is a Fr\'echet space when equipped with the graph topology $\tau_{dU}$. 

The graph topology $\tau_{dU}$ on $\mathcal{D}^{\infty}(U)$ can be generated by other families of semi-norms. 
\begin{lemma}[Schm\"udgen \cite{xxx}, Corollary 10.2.4]
Let a be an elliptic element of $ \mathcal{E}(\mathcal{G}) $, then $\mathcal{D}^{\infty}(U) = \mathcal{D}^{\infty}(\overline{dU(a)})$ and the graph topology $\tau_{dU}$ on $\mathcal{D}^{\infty}(U) $ is generated by the family of semi-norms $\| \cdot \|_{dU(a)^{n}}$, $n \in \NN_{0}$.
\end{lemma}
 
The two families of semi-norms $\{ \| \cdot \|_{dU(a)^{n}}, \; n \in \NN_{0}\}$ for an elliptic $a \in \mathcal{E}(\mathcal{G}) $ and $\{\| \cdot \|_{dU(y)}, \; y \in  \mathcal{E}(\mathcal{G}) \}$ are equivalent because they both generate the graph topology on $\mathcal{D}^{\infty}(U)$ and the first is clearly a directed family so for each $y \in  \mathcal{E}(\mathcal{G})$ there exists a constant $K$ and an integer $n\in \NN$ so that for all $\phi \in \mathcal{D}^{\infty}(U)$, $\|dU(y)\phi\| \leq K \|dU(a)^{n}\phi \|.$

\subsection{The state spaces $ \EsubS(\mathcal A)$\ } 

Let the algebra of physical qualities  be represented by the $O^{\ast}$-algebra 
$\mathcal M$ defined on the dense subset $\mathcal{D} \subset \scrH$. We assume that $\mathcal M$ has a unit element, the identity operator $\hat I$.

A linear functional $f$ on the $O^{\ast}$-algebra $\mathcal A$ is a linear map from $\mathcal A$ 
to the standard complex numbers $\mathbb C$, it is strongly positive iff 
$f(\hat X)\geq{0}$ for all $\hat X\geq{0}$ in $\mathcal A$. 

\begin{definition}
The states on $\mathcal A $ are the strongly positive 
linear functionals on $\mathcal A$ that are normalised to take the 
value 1 on the unit element $\hat I $ of $\mathcal A $. The state space $\EsubS (\mathcal A)$ of the $O^{\ast}$-algebra $\mathcal A$ is the set of all states on $\mathcal A$.
\end{definition}
\begin{definition}
A bounded operator $\hat B$ on $\scrH$ is trace class iff $Tr |\hat B| < \infty$.
A trace functional on $\mathcal M$ is a functional of the form $ \hat A \in \mathcal A \mapsto Tr (\hat B \hat A) $ for some trace class operator $\hat B$.
\end{definition}
The states of the $O^{\ast}$-algebra $\mathcal A$ are given by trace class operators.
\begin{theorem}\cite{xxx}\label{TH1}
Every strongly positive linear
functional  on $\mathcal A $ is given by a trace functional.
\end{theorem}

The state space $\EsubS(\mathcal A)$ is contained in the convex
hull of projections $\scrP$ onto one-dimensional subspaces
spanned by unit vectors $\phi  \in \mathcal D$. If 
$\hat \rho \in \EsubS (\mathcal A)$ then there is an orthonormal set of vectors 
$\{\phi_{n}\}_{n\in \NN'}$  in $\mathcal D$  such that
$\hat \rho = \sum_{n \in \NN'}\lambda_{n}
|\phi_{n}\rangle\langle\phi_{n}| = 
\sum_{n \in \NN'}  \lambda_{n} \scrP_n $ where
$\scrP_n = |\phi_{n}\rangle\langle\phi_{n}|$ is the orthogonal projection onto the
one-dimensional subspace spanned by $\phi_{n}$,
$\lambda_{n} \in \RR, 0 \leq \lambda_{n} \leq 1 $ and $\sum_{n \in \NN'}
 \lambda_{n} = 1$ with $\NN' = \{ n \in \NN; \lambda_{n} \neq 0 \}$.  All states
in $\EsubS(\mathcal A)$ satisfy the condition that as $n$ approaches
infinity the sequence $\{\lambda_{n}\}$ converges to zero faster than any
power of $1/n$ \cite{adelman2}. 

The following well-known example is our guide. For the $C^{\ast}$-algebra $ \scrB( \scrH)$, the state space $\scrE = \EsubS(\scrB( \scrH))$ is composed of operators in $\scrT_{1}(\scrH)$ that are self-adjoint and normalised. The collection $\scrT_{1}(\scrH)$ of all trace class operators on $\scrH$ is a linear space over $\scrC$, it is a Banach space when equipped with the trace norm  $\nu(T) = Tr |T|$, where $|T| = \surd (T^{\ast} T)$. The open subsets of $\scrE$ in the trace norm topology are denoted  
$\nu( \rho_{1} ; \delta) = \{  \rho \; :Tr | \rho -  \rho_{1}| < \delta\}$. $\scrE$ is compact in the weak$^\ast$ topology, the weakest topology on  $\scrE$ that makes continuous all the functionals $ \rho \to Tr ( \rho\hat B) , \hat B \in  \scrB( \scrH)$. Its sub-basic open sets are $ \scrN( \rho_{0}  ; \hat B ; \epsilon) = \{\rho \; : |Tr \rho \hat B - Tr \rho_{0} \hat B| < \epsilon \}$ with $\hat B \in \mathcal B(\mathcal H)$ and $\epsilon > 0$.

In order to carry out a similar analysis for the states on an $O^{\ast}$-algebra $\mathcal A$, we define a subset $\mathcal{T}_{1}(\mathcal A)$ of $\scrT_{1}(\scrH)$. Let $\mathcal{T}_{1}(\mathcal A) \subset \scrT_{1}(\scrH)$ be composed of trace class operators $\hat T$ that satisfy $ \hat T\scrH \subset \mathcal D , \hat T^{\ast}\scrH \subset \mathcal D $ and $\hat A \hat T, \hat A \hat T^{\ast} \in
\mathcal{T}_{1}(\scrH)$ for all $\hat A \in \mathcal A $, that is,
$\mathcal{T}_{1}(\mathcal A)$ is a $\ast$-subalgebra of $\mathcal A$ satisfying $\mathcal A
\mathcal{T}_{1}(\mathcal A) = \mathcal{T}_{1}(\mathcal A)$ \cite{inoue}. 
\begin{definition}
The state space $\EsubS(\mathcal A)$ for the $O^{\ast}$-algebra $\mathcal A$ is the set of normalized, self-adjoint operators in $\mathcal{T}_{1}(\mathcal A)$.
\end{definition}
With parameters $\hat A \in \mathcal A, \epsilon > 0$ and $\rho_0 \in \EsubS(\mathcal A)$, the sets  $ \scrN( \rho_{0}  ; \hat A ; \epsilon) = \{
\rho \ ; |Tr \rho \hat A - Tr \rho_{0} \hat A| < \epsilon \}$  form an open 
sub-base for the weak topology on $\EsubS(\mathcal A)$ generated by the functions $ a_{Q}(\cdot)$. 
The weak topology on $\EsubS(\mathcal A)$ is stronger than the weak$^\ast$ topology on $\scrE$ restricted to $\EsubS(\mathcal A)$ because for any $\hat B \in \scrB( \scrH)$ the weak$^\ast$ sub-basic open set  $ \scrN( \rho_{0}  ; \hat B ; \epsilon)$ is also open in the weak topology \cite{reed}. 

It is well known that for all $\rho_1 \in \EsubS(\mathcal A)$ and every $\delta
> 0$, if $\hat A \in \scrB(\scrH)$ then $\nu( \rho_{1} ; \delta) \subset  \scrN(
\rho_{1}  ; \hat A ;  K \delta) $ where $K = \| A \| $. The following
generalisation uses the precompact topology on $ \mathcal A$ determined by the
family of semi-norms $ \mathcal P_{ M, N} (\hat A) = \sup_{\xi \in M, \eta \in
N}
| \langle \hat A \xi, \eta \rangle | $ where $\{ M, N \}$ range over subsets of $\mathcal D$ precompact in the graph topology $t_{\mathcal A}$, see Schm\"udgen \cite{xxx}, Sections 5.3. A subset of $\mathcal D$ is precompact in the graph topology  $t_{\mathcal A}$ if its closure is compact in $t_{\mathcal A}$. 

\begin {lemma} If $\mathcal D$ is a Fr\'echet space in the topology $t_{\mathcal A}$, then every essentially self-adjoint operator $\hat A \in \mathcal A $ defines a linear functional $ G_{\hat A}(\hat T) = Tr \hat A \hat T$ on $\mathcal{T}_{1}(\mathcal A)$ such that for every self-adjoint  
$\hat T \in \mathcal{T}_{1}(\mathcal A)$, 
\begin{equation}
|G_{\hat A}(\hat T)| = |Tr \hat A \hat T| \leq   p_{\hat T}(\hat A) \nu (\hat T).
\end{equation} 
$\mathcal P_{\hat T} (\hat A) = \sup_{\zeta_{n}} |\langle \hat A \zeta_{n},\zeta_{n} \rangle|  \leq \sup_{\zeta_{n}} \| \hat A \zeta_{n} \|$, where the set $\{\zeta_{n} \in \mathcal{D}\}$ is an orthonormal set of eigenvectors of $|\hat T|$ for eigenvalues $\lambda_{n} > 0 $ that includes the  orthonormal bases $\{\zeta_{m}\}_{m=1}^{s}$ of eigenspaces for $\lambda_{n} $ with multiplicity $s> 1$.
\end{lemma}  
Note that $p_{\hat T}(\hat A)$ depends on $\hat T$, if it was independent of $\hat T$ then the inequality would show that  $G_{\hat A}$ was continuous with respect to the trace norm topology on $ \mathcal{T}_{1}(\mathcal A)$. 
  
\begin{proof}
Since $\mathcal D$ in the topology $t_{\mathcal A}$ is metrizable, a result of Grothendieck \cite{kothe} shows that the operator $\hat T\in \mathcal{T}_{1}(\mathcal A)$ has a canonical representation $\hat T = \sum _{n} \lambda_{n}  \langle \;, \zeta_{n} \rangle \eta_{n}$, where $\sum _{n}| \lambda_{n} | < \infty$ and the sequences $( \zeta_{n} )$ and $( \eta_{n} )$ converge to zero in $\mathcal D$. 
$\sum _{n} \lambda_{n}  \langle \;, \zeta_{n} \rangle \eta_{n}$ converges absolutely with respect to $\mathcal A $, that is, $\sum _{n} \lambda_{n}  \| \hat A   \eta_{n}\| \|\hat C \zeta_{n}\| < \infty$ for all $\hat A, \hat C \in \mathcal{A}$. 

If $\hat T$ is self-adjoint it has a canonical representation $\sum _{n} \lambda_{n}\langle \;,  \zeta_{n} \rangle \zeta_{n}$, the $\{\zeta_{n}\}$ being an orthonormal set of vectors in $\mathcal{D}$ where $\zeta_{n}$ is an eigenvector of $|\hat T|$ for the eigenvalue $\lambda_{n} > 0 $. By the spectral theorem, the series $\Sigma _{n} \lambda_{n}  \langle \;,\zeta_{n} \rangle \zeta_{n}$ is strongly convergent in $\scrH$ and $\Sigma_{n} | \lambda_{n}| \|\hat A \zeta_{n}\|^{2} < \infty$ for any 
$\hat A \in \mathcal{A} $ \cite[Lemma 5.1.10]{xxx}. Therefore the set $\{ \zeta_{n} \}$ has at most the single limit point $0 \in \mathcal D$ in the graph topology.
  
For a fixed $\hat T \in \mathcal{T}_{1}(\mathcal A)$ and any $\hat A \in \mathcal{A} $,
\begin{equation}
 |Tr \hat A \hat T | = | \Sigma _{n} \langle (\hat A \lambda_{n} \zeta_{n}), \zeta_{n} \rangle | \leq (\Sigma_{n} |\lambda_{n}|) \mathcal P_{ \hat T} (\hat A).
\end{equation}
The set $V(\hat T) =  \{ \zeta_{n} \}$ of eigenvectors of $|\hat T|$ form a
pre-compact set in $\mathcal D (t_{\mathcal A})$ so that 
\[\mathcal P_{\hat T} (\hat A) = \sup_{\zeta_{n}} |\langle \hat A
\zeta_{n},\zeta_{n} \rangle| \leq \sup_{\zeta_{n}} \| \hat A \zeta_{n} \|\]
because $\| \zeta_{n}\| = 1$ for all $n$. $\nu (\hat T) = (\Sigma_{n}
|\lambda_{n}|) $ is the trace norm of $\hat T$. We cannot immediately infer that
$G_{\hat A}$ is continuous with respect the trace norm restricted to
$\mathcal{T}_{1}(\mathcal A)$ because $ \mathcal P_{\hat T} (\hat A)$ depends on
$\hat T$. If the supremum was taken over all of $\mathcal{D}$ the argument could
be extended to a proof of the desired continuity of  $G_{\hat A}$.
\end{proof}
There are two special classes of self-adjoint operators $\hat A$ for which we can prove continuity of $G_{\hat A}$ with respect the trace norm restricted to $\mathcal{T}_{1}(\mathcal A)$. In the first, the operator $\hat A$ is bounded on $\mathcal{H}$ then $ \mathcal P_{ \hat T} (\hat A) \leq \sup_{\zeta_{n}} \| \hat A \zeta_{n} \| \leq  \|\hat A \|$. The second class contains operators $\hat A \in \mathcal{A}$ which are `relatively bounded' with respect to a positive self-adjoint operator $\hat N \in \mathcal{A} $ with $\hat N + \hat I$ an invertible operator that maps $\mathcal{D}$ into itself. Relative bounded means that  $\hat A (\hat N + \hat I)^{-1}$ is a bounded operator. Then $ \mathcal P_{ \hat T} (\hat A) \leq
\sup_{\zeta_{n}} \| \hat A \zeta_{n} \| \leq \sup_{\phi \in \mathcal{D}} \| \hat A \phi \| \leq \sup_{\psi \in \mathcal{D}} \| \hat A (\hat N + \hat I)^{-1}\psi \| \leq  \|\hat A (\hat N + \hat I)^{-1}\|$. 

\subsection{When $\mathcal{A}$ is  $dU(\mathcal{E}(\mathcal{G}))$}
If the $O^{\ast}$-algebra $\mathcal{A}$ is the infinitesimal representation $dU$ of the enveloping algebra $\mathcal{E}(\mathcal{G})$ obtained from a unitary representation $U$ of $G$. Take the positive self-adjoint operator $\hat N$ to be $  - \Delta$ where $\Delta = \sum_{i=1}^{d} x_{i}^{2}$ is the Nelson Laplacian in the enveloping algebra of the Lie algebra $\mathcal{G}$ with basis $\{ x_{1},x_{2},......,x_{d}\}$.

\begin{theorem} \label{TH2}
Each essentially self-adjoint operator $\hat A = dU(x)$ that represents an element $x \in \mathcal{E}(\mathcal{G}) $ defines a linear functional $G_{\hat A}$ on 
$\mathcal{T}_{1}(dU(\mathcal{E}(\mathcal{G}))$ that is continuous with respect to the trace norm topology on $ \mathcal{T}_{1}(dU(\mathcal{E}(\mathcal{G}))$. 
\end{theorem}
\begin{proof} 
Write $\mathcal A$ for $dU(\mathcal{E}(\mathcal{G}))$.
From the preceding lemma we have that for a fixed $\hat T \in \mathcal{T}_{1}(\mathcal A)$ and any $\hat A \in \mathcal{A}$,
\begin{equation}
 |Tr \hat A \hat T | \leq (\Sigma_{n} |\lambda_{n}|) \mathcal P_{ \hat T} (\hat A) \leq \nu (\hat T)\sup_{\phi \in \mathcal{D}^{\infty}(U)} \| \hat A \phi \|.
\end{equation}
The graph topology on $\mathcal{D}^{\infty}(U)$ is generated by both families of semi-norms $\{ \|\cdot\|_{dU(x)} = \|dU(x)\cdot\|\; ; x \in \mathcal{E}(\mathcal{G}) \}$ and $\{ \| \cdot \|_{dU(1 - \Delta)^{n}}, \; n \in \NN_{0} \}$ hence for each $y \in \mathcal{E}(\mathcal{G})$ there exists $m \in \NN$ and $K\in \RR$ with $\|dU(y)\phi\| \leq K \|dU((1- \Delta)^{m}) \phi \|$
for all $\phi \in \mathcal{D}^{\infty}(U) =  \mathcal{D}^{\infty}(\overline{dU(1 - \Delta)})$, and $\mathcal{D}^{\infty}(\overline{dU(1 - \Delta)}) = \bigcap_{n\in\NN} \mathcal{D}(\overline{dU((1 - \Delta)}^{n}))$.

For any integer $m \in \NN$ the operator $dU((1- \Delta)^{m})= (dU(1- \Delta))^{m}$ is a positive, essentially self-adjoint operator on $\mathcal{D}^{\infty}(U) $ that maps  $\mathcal{D}^{\infty}(U) $ onto itself, so that for each $\phi \in \mathcal{D}^{\infty}(U)$ there exists $\psi \in \mathcal{D}^{\infty}(U)$ such that $\phi = dU((1- \Delta)^{m})^{-1} \psi$. Therefore for each $y \in \mathcal{E}(\mathcal{G})$,
$\|dU(y)\phi\| = \|dU(y)dU((1- \Delta)^{m})^{-1}  \psi\| \leq K \|\psi\|.$
Put $\hat A = dU(y)$ and normalize $\psi$, so that
$ |Tr \hat A \hat T | \leq \nu (\hat T)\sup_{\psi \in \mathcal{D}^{\infty}(U)} \| \hat A dU((1- \Delta)^{m})^{-1} \psi \|/ \|\psi\| \leq \nu (\hat T)  K .$
 
Let $ \rho_1, \rho \in \EsubS(\mathcal{E}(\mathcal{G}))$ and put $ \hat T = ( \rho -   \rho_1)$ and $K = K(\hat A)$, which is independent of $\hat T$, then 
\begin{equation}
|Tr (\hat A \rho - \hat A \rho_1)| \leq  K(\hat A) Tr | \rho -  \rho_1|.
\end{equation} 
\end{proof}
Note that if $\hat A$ represents an element $y \in \mathcal{G}$ then we can take $m =2$ and $m = 4$  if $y$ is a quadratic element of  $\mathcal{E}(\mathcal{G})$. 

\begin{corollary}\label{CR7}
Each function $a_{Q}(\cdot)$ for an essentially self-adjoint operator $\hat A \in dU(\mathcal{E}(\mathcal{G}))$  is continuous from $\EsubS(dU(\mathcal{E}(\mathcal{G})))$ to $\RR$ in the trace norm topology restricted to  $\EsubS(dU(\mathcal{E}(\mathcal{G})))$. 
\end{corollary}
\begin {corollary}\label{CR8}
The open sets  $\{\nu ( \rho_{1} ; \delta);  \rho_{1} \in \EsubS(dU(\mathcal{E}(\mathcal{G}))), \delta > 0 \}$ form an open basis for the weak topology on $\EsubS(dU(\mathcal{E}(\mathcal{G})))$.
\end{corollary}
\begin {corollary}\label{CR9}
The weak topology on $ \EsubS(dU(\mathcal{E}(\mathcal{G})))$ is Hausdorff.
\end{corollary} The next results determine the interior of sets of states satisfying equations of the form $Tr \rho \hat A = \alpha$ .

\begin{lemma}\label{LM9}
Given $ \rho \in  \EsubS (dU(\mathcal{E}(\mathcal{G})))$ with $Tr \rho \hat A  = 0$ then for all $\epsilon > 0$ there we can construct states  $ \sigma \in  \EsubS(dU(\mathcal{E}(\mathcal{G})))$ such that $ \sigma \in \nu( \rho ; \epsilon)$ and $Tr \sigma \hat A \ne 0$.
\end{lemma} 
\begin{proposition} \label{PR10}
Let $\hat A \in dU(\mathcal{E}(\mathcal{G}))$ be a non-zero essentially self-adjoint operator, then
 $\text{int}\{ \rho | Tr  \rho \hat A = 0 \} = \emptyset$ where the interior is taken with respect to the weak topology on $\EsubS(dU(\mathcal{E}(\mathcal{G})))$.
\end{proposition}
\begin{corollary}\label{CR11}
For any standard real number $\alpha$, $\text{int}\{ \rho | Tr \rho \hat A = \alpha \} = \emptyset$.
\end{corollary}

\subsection { Quantum real numbers}

\begin{definition} We say that a quantum real number (qr-number) is a section of the sheaf of Dedekind reals $\RsubD(\EsubS(\mathcal A))$ in $\Shv (\EsubS(\mathcal A))$. $\RsubD(\EsubS(\mathcal A))$ is isomorphic to $C(\EsubS(\mathcal A))$, the sheaf of germs of continuous functions on $\EsubS(A)$.
\end{definition} For each non-empty $U \in \mathcal{O}( \EsubS(\mathcal{A}))$, the subsheaf $\RsubD(U)$: (1) has integers $\ZZ(U)$, rationals 
$\QQ(U)$ and Cauchy reals $\RsubC(U)$ as subsheaves of 
locally constant functions, (2) has
orders $<$ and  $\leq$ compatible with those on $\QQ(U)$  but $<$ is not total because trichotomy, $x>0 \lor x=0 \lor x<0$, is not satisfied. $\leq$ is not equivalent to $< \lor =$,  (3) is closed under the commutative, associative,
distributive binary operations  $+$  and $\times$, has $0 \ne 1$
and is a residue field, i.e., if $ b \in \RsubD(U)$ is not invertible then $b = 0$ and (4) has a distance 
function $|\cdot|$ which defines a metric with respect to which it is a complete 
metric space in which $\QQ(U)$ dense. 
A section $b\in \RsubD(U)$ is apart from $0$ iff  $|b| > 0$. $\RsubD(U)$ is an apartness field, i.e., $\forall b \in \RsubD(U), \;  |b| > 0 $ iff $b$ is invertible.

\subsubsection{Locally linear qr-numbers}Every qr-number is a continuous real function of locally linear qr-numbers $\AAA(\EsubS(\mathcal A))$ defined as locally linear functions on the state space $\EsubS(\mathcal A)$.
\begin{definition}\cite{adelman1}
Let $U$ be an open subset of $\EsubS(\mathcal A)$.  A  function $f : U \to \RR$, the 
standard reals, is {\em locally linear} if each $\rho \in U$ has an open neighborhood $U_{\rho} \subset U$ with an essentially self-adjoint operator $\hat A \in \mathcal A$ such that $f|_{U_{\rho} } = a_{Q}(U_{\rho} )$.
\end{definition}The global elements of $\AAA(\EsubS(\mathcal A))$ are given by the functions $a_{Q}$.
 The set $\AAA(\EsubS(\mathcal A))$ is dense in $\RsubD(\EsubS(\mathcal A))$.

On any non-empty open set $W$, $a_Q(W)
= b_Q(W)$ if and only if the defining operators are equal, $\hat A = \hat
B$. Therefore knowing $a_Q(W)$ on any open set $W \neq \emptyset$ is equivalent to knowing the operator $\hat A$.

We will need the following properties of  locally linear qr-numbers.

\subsubsection{ $\epsilon$ sharp collimation of a locally linear qr-number}
The condition $W$ for $\epsilon$ sharp collimation of an attribute $\hat A$ in an interval $I_{a}$ with midpoint $a_{0}$ and width $|I_{a}|$  is that  for $0 < \epsilon < 1$, both
$ |a_{Q}(W) - a_{0}| \leq \epsilon |I_{a}|/2 $ and $ a_0 -  |I_{a}|/2\ \leq\ a_{Q}(W) -\frac {s(a)_{Q}(W)}{ \sqrt{\epsilon}}
\ < a_{Q}(W) +\frac {s(a)_{Q}(W)}{ \sqrt{\epsilon}} \ \leq\ a_{0} +  |I_{a}|/2$
where $s(a)_{Q}(W) = \sqrt{(a^{2}_{Q}(W)  - a_{Q}(W)^{2})}$. Hence $a_{Q}(W) \in I_{a}  \; \text{and}  \; (a^{2}_{Q}(W) - a_{Q}(W)^{2}) \leq  \frac{\epsilon}{4}|I_{a}|^{2}$ so that the qr-number $a_{Q}(W)$ is equal to the constant qr-number $a_{0} 1_{Q}(W)$ with an accuracy less than or equal to $\epsilon$. 
  
\subsubsection{ $\epsilon$-location in subintervals of $\RR$}
The $\epsilon$-location of $\hat A \in\mathcal{A}$ in a sub-interval $I_{a}$ of the numerical range of $\hat A$ is defined using $\hat P^{\hat A}(I_{a})$, the spectral projection operator of $\hat A$ on
$I_{a}$ which has the qr-number value $\pi^{\hat A}(I_{a})_{Q}(W)$  in the condition $W$.

\begin{definition} In the condition $W$, $\hat A$ is 
$\epsilon$-located in the interval $I_{a}$ if
\begin{equation}
(1 - \epsilon) < \pi^{\hat A}(I_{a})_{Q}(W)  \leq  1.
\end{equation}
\end{definition}
The concepts of $\epsilon$-location and $\epsilon$ sharp collimation in an
interval $I_{a}$ are closely related.
\begin{theorem}\cite{durt2} If the collimation of an attribute $\hat A$ of a quantum system in an interval $I_{a}$ is
$\epsilon$ sharp on the open set $W$, then $\hat A$ is $\epsilon$-located in $I_{a}$ in the condition $W$.\end{theorem}
\subsubsection{Strictly $\epsilon$ sharp collimation.}
\begin{definition}
Let $ \hat P^{\hat A}(I_{a})$ be the spectral projection operator of $\hat A$ on $I_{a}$, then $\hat A$ is strictly $\epsilon$ sharp collimated on $I_{a}$ in $W$ if, in addition to being $\epsilon$ sharp collimated on $I_{a}$ in $W$, $ Tr | \rho - \hat P^{\hat A}(I_{a}) \rho \hat P^{\hat A}(I_{a}) | < \epsilon$ for all $\rho \in W$.
\end{definition}  
\begin{theorem} If $\rho_{m} = |\psi_{m}\rangle\langle \psi_{m}|$ with $ \hat A\psi_{m} = a_{0}\psi_{m}$, e.g.  $\rho_{m}$ is a pure eigenstate of $\hat A$, such that $\forall \epsilon > 0, \exists 0< \delta < \epsilon/3 $ with $Tr|\hat P^{\hat A}(I_{a}) \rho_{m}\hat P^{\hat A}(I_{a})  - \rho_{m}| < \delta$, then if $\hat A$ is $\epsilon$ sharp collimated on $I_{a}$ in $\nu(\rho_{m}, \delta)$ then it is strictly $\epsilon$ sharp collimated in $I_{a}$ on $\nu(\rho_{m}, \delta)$.
\end{theorem}\begin{proof} $\forall \rho \in \nu(\rho_{m}, \delta)$ ,
 $Tr | \rho - \hat P^{\hat A}(I) \rho \hat P^{\hat A}(I) | \leq Tr | \rho - \rho_{m}| + Tr | \rho_{m} - \hat P^{\hat A}(I)\rho_{m} \hat P^{\hat A}(I)| + Tr | \hat P^{\hat A}(I)( \rho  - \rho_{m})\hat P^{\hat A}(I)| < 3 \delta < \epsilon$
because $Tr | \hat P^{\hat A}(I)( \rho  - \rho_{m})\hat P^{\hat A}(I)| \leq  Tr |( \rho  - \rho_{m})|$. The definition of the polar decomposition of bounded operators implies that $Ker  | \hat P^{\hat A}(I)( \rho  - \rho_{m})\hat P^{\hat A}(I)| = Ker  \hat P^{\hat A}(I)( \rho  - \rho_{m})\hat P^{\hat A}(I) \supseteq Ker (\rho  - \rho_{m}) = Ker | (\rho  - \rho_{m})|$.
\end{proof}

\subsubsection{ The extended qr-numbers}
Extended qr-numbers are just extended Dedekind real numbers $\RsubD^{0}(\EsubS(\mathcal M))$ so that an extended locally linear qr-number is a prolongation by zero \cite{swan} of a locally linear qr-number $a_Q(W)$ defined on the open set $W \subset  \EsubS(\mathcal M)$. The prolongation by zero of a qr-number that is a continuous function of locally linear functions defined on a common open set $W$ is given by  the same function of the prolongation by zero of each of the locally linear functions.

In the following, the discussion is in terms of qr-numbers, this term may be also used for extended  qr-numbers, however when this occurs the meaning should be clear from the context.

\subsubsection{qr-numbers generated by a single operator}

If $a_{Q}$ is defined by an essentially self-adjoint operator $\hat A \in \mathcal{A}$, then we can define the sheaf of locally linear qr-numbers $\AAA^{a_{Q}}$ and the sheaf of qr-numbers, $\RsubD^{a_{Q}}(\EsubS(\mathcal A))$, generated by $\hat A$. For $\hat A \in \mathcal{A}$ and each non-empty open subset $U \in \mathcal{O}(\EsubS(\mathcal A))$ there is a locally linear qr-number $a_{Q}(U)$ and for each continuous function $F : \RR \to \RR$ a qr-number $F(a_{Q}(U))$ defined to extent $U$. Then $\AAA^{a_{Q}}(\EsubS(\mathcal A))$ is a sub-sheaf of $\RsubD^{a_{Q}}(\EsubS(\mathcal A))$ and both are sub-sheaves of $\RsubD(\EsubS(\mathcal A))$. By construction both contain the integers $\ZsubD(\EsubS(\mathcal A))$, rationals $\QsubD(\EsubS(\mathcal A))$ and Cauchy reals $\RsubC(\EsubS(\mathcal A))$ as subsheaves of locally constant functions. Thus $\RsubD^{a_{Q}}(\EsubS(U))$ contains all the standard real number $\alpha(U)$ defined as constant functions on $U$. The qr-numbers  $\RsubD^{a_{Q}}(\EsubS(U))$ can be extended by zero to be globally defined. 

The order relations and the distance function defined on $\RsubD(\EsubS(\mathcal A))$ restrict to 
$\RsubD^{a_{Q}}(\EsubS(\mathcal A))$ and proofs of Stout \cite{stout} that use the properties of Dedekind cuts in $\QsubD(\EsubS(\mathcal A))$ hold for $\RsubD^{a_{Q}}(\EsubS(\mathcal A))$ so that,
\begin{proposition}\label{PR28}
When the metric topology $T$ is restricted to $\RsubD^{a_{Q}}(\EsubS(\mathcal A))$ the sheaf of rational numbers  $\QsubD(\EsubS(\mathcal A))$ is dense in $\RsubD^{a_{Q}}(\EsubS(\mathcal A))$.
If $\hat A$ is not a constant  operator the quantum real numbers $\RsubD^{a_{Q}}(\EsubS(\mathcal A))$ form a complete metric space.
\end{proposition}

\begin{proposition}\label{PR29}
A Cauchy sequence of locally linear quantum real numbers 
$\{ a_{Q}(V_{k})\} \in \AAA^{a_{Q}}$, for a fixed operator $\hat A$, converges to a locally linear quantum real number $a_{Q}(W)$ for some open subset $W$.\end{proposition}

The algebraic properties of $\AAA^{a_{Q}}$ are inherited from those of $\AAA$. When the operator satisfies $\hat A > 0 \lor \hat A < 0 $ then $a_{Q}$ is invertible, and $\AAA^{a_{Q}}$ is an apartness field. It is also a residue field.

 Between any pair of rational numbers $r < s$ there exists a locally linear quantum real number in $\AAA^{a_{Q}}$, namely, $a_{Q}(W)$ where $W = \{ \rho : r < Tr \rho \hat A < s\}$. Furthermore if the numerical range of $\hat A$ is $\RR$ then the quantum real numbers $\RsubD^{a_{Q}}(\EsubS(\mathcal A))$ constitute a Dedekind real numbers object. With this real number continuum we will be able to define functions with values in $\RsubD^{a_{Q}}(\EsubS(\mathcal A))$ and use an integral and differential calculus of these functions.
 
\subsection{Some conceptual enigmas viewed from a qr-numbers perspective. }

\subsubsection{Statistics:  ontological and epistemological conditions}
  
The physical interpretation of the conditions of a quantum system builds upon the interpretation \cite{jauch} that the state of a quantum system is the result of the preparation of the system.  In the qr-number model the condition of a system encapsulates its prior history. This history may include experimental preparations as well as non-experimentally controlled interactions. Each system always has an ontic  condition which allows it to have a deterministic evolution.
\begin{definition} The ontological conditions of the system whose qualities are represented by operators in $dU(\mathcal{E}(\mathcal{G})))$ are given by the open sets $W$ in the weak topology on $ \EsubS(dU(\mathcal{E}(\mathcal{G})))$.
\end{definition}
 If a system is experimentally prepared in an open set $W$ of state space, then the epistemological condition of the system is $W$.  The ontological condition of the system can be any non-empty open set $V \subset  W$ because the values of qualities defined to extent $V$ will satisfy the experimental restrictions imposed on qualities defined to extent $W$. Usually the epistemic condition holds for an ensemble of identically prepared systems, while any member of the ensemble will have an ontic condition $V \subset  W$. This leads to an ignorance interpretation for the statistics.  We have shown \cite{durt2} that if the epistemological condition for an ensemble of systems is the open set $\nu(\rho_{0}, \delta)$ for positive $\delta \ll 1$ then the outcomes of a dichotomic experiment are well approximated by expectation values evaluated at the quantum state $\rho_{0}$ from which Born's quantum probability rule is obtained, see $ \S II$ of \cite{durt2}.
 
\subsubsection{Measurement, approximating qr-numbers by classical real numbers}
The relation of the invisible qr-numbers to observable standard real numbers is elucidated through the single slit experiment, using the fact that the standard real numbers can be realised as locally constant qr-numbers. We will use the concept of $\epsilon$ sharp collimation of an attribute $\hat Z$ in a slit $I_{z}$, with midpoint $z_{0}$ and width $|I_{z}|$, when the system has the condition $W$. 
For $0 < \epsilon < 1$,  $\hat Z$ is  $\epsilon$ sharp collimation in the slit  $I_{z}$ if  $z_{Q}(W) \in I_{z}  \; \text{and}  \; (z^{2}_{Q}(W) - z_{Q}(W)^{2}) \leq  \frac{\epsilon}{4}|I_{z}|^{2}$.  That means that $z_{Q}(W)$ equals the standard real number $z_{0}$ on $W$ with an error less than or equal to $\epsilon$. 

The $\epsilon$ sharp collimation  of $\hat A$ depends upon the spectrum of $\hat A$.
The construction of the open sets $W\in \mathcal{O}(\EsubS(\mathcal M))$ on which $\hat A$ is $\epsilon$ sharp collimation in $I $ requires $\hat A$ to have spectrum in $I$ \cite{JVC2}.

If $W$ is an epistemological condition of the system, in any run of the experiment the system will have an ontic condition $U$, if $U\cap W \neq \emptyset $ then the attribute $\hat Z$ registers the measured value $z_{0}\in \RR$, if $U \cap W = \emptyset$ then $\hat Z$ doesn't. A system in different ontic conditions may have observationally different outcomes for attributes other than $\hat Z$. 

The measurement problem arises because quantum mechanics predicts only a probability distribution of the values obtained by measuring a physical quantity on an ensemble of systems which are all prepared identically.  But probabilities can only be determined if each outcome can be observationally distinct.  This is not obtained in the standard QM description for which the final reduced state of both the system and measurement apparatus is a mixed state. 

\subsubsection{Locality in qr-number space}
Our understanding of the geometry of physical space starts from Riemann, who
isolated two hypotheses in his 1854 lecture \cite{riemann} entitled {\em On the
hypotheses which lie at the foundations of geometry}: 1. A topological
hypothesis: locations are fixed by allocating multiplets of real numbers. The
possibility of using different real number systems was not explored. 2. A
metrical hypothesis: the distance between located points is given by a metric
function.

Riemann assumed that the metric was not given once and for all but had to be physically determined because it was ``causally connected with matter''. This led to Riemannian geometries.   For the topology, we assume that the real number system is not given once and for all but must also be ``physically determined'', leading to Dedekind reals in different topoi.

For example, if the algebra of attributes for a single massive Galilean relativistic quantum particle is
 $\mathcal{A} = dU(\mathcal{E}(\mathcal{G})))$ then the quantum space is $\RsubD(\EsubS(\mathcal{A}))^{3}$ \cite{JVC3}.  The cartesian coordinate axes are parametrized by $\AAA^{x^{j}_{Q}}(\EsubS(\mathcal{M}))$ generated by the position operators $\hat X^{j}$, for $j= 1,2,3$, that  transform appropriately under the Euclidean subgroup of the symmetry group $G$.   The triplets $\vec x_{Q}(W)$ for $W \in \mathcal{O}(\EsubS(\mathcal{A}))$ do not label points, they are open sets because any section $a_{Q}(W)$ is open in $\RsubD(\EsubS(\mathcal{A}))$ by the construction of its topology. They are like Russell's events \cite{russell}. The graph of $\vec x_{Q}(W)$ is $(W, O_{\vec x})$ with $ O_{\vec x}$ an open subset of the standard Euclidean space $\RR^{3}$. 

\subsubsection{Evidence that the spatial continuum of quantum phenomena is different from the classical.}
(a) In the Einstein-Podolsky-Rosen-Bohm-Bell experiments for two spin one-half massive particles, the particles are prepared so that the sums of their momenta and their spins are both zero. They are sent to two Stern-Gerlach apparatuses, $B_{R},\;B_{L} $, a large distance apart, whose magnetic fields can be set independently in directions $\vec u(R),\; \vec u(L)$.

If quantum space is assumed to be classical and each particle is assumed to arrive at one of  $B_{R},\;B_{L} $ then the experiment contravenes Einstein locality because changing the direction of the magnetic field in one apparatus affects the particle in the other.
  
The qr-numbers approach maintains Einstein locality in the quantum space of the entangled two particle system because the single particles are always close to each other in their qr-number space and hence can always interact \cite{JVC1}. It gives the usual quantum mechanical results for the experiment. In a suitably prepared condition the qr-number trajectory for each particle goes both to $B_{L}$ and up along $ \vec u_{L}$ and to $B_{R}$ and down along $\vec u_{R}$ to extent  $\tilde W_{L} ^{+} \cup \tilde W_{R} ^{-}$, and each goes both to $B_{L}$ and down along $ \vec u_{L}$ and to $B_{R}$ and up along $\vec u_{R}$ to extent $\tilde W_{L} ^{-} \cup \tilde W_{R} ^{+}$. A particle with ontic condition $V_{n}$ won't register in a detector in the upper sector of $B_{R}$ unless $V_{n} \cap \tilde W_{R} ^{+} \ne \emptyset$. 

(b) Double slit experiments. A single particle localised in classical space cannot simultaneously pass through two spatially separated slits. Nevertheless the experimental evidence \cite{pworld,tono,merli} clearly shows single electrons building up an interference pattern. This building up process is described deterministically in the qr-number model \cite{JVC2}. 

In the qr-numbers interpretation a quantum particle can simultaneously pass through two slits $I_{z}^{+},\; I_{z}^{-}$ that are separated in classical space, $I_{z}^{+}\cap I_{z}^{-} = \emptyset $, when its single location in qr-number space is an open set that is the union of disjoint parts $z_{Q}(V_{+})$ and $z_{Q}(V_{-})$ with $z_{Q}(V_{+}) \subset I_{z}^{+}$ and $z_{Q}(V_{-}) \subset I_{z}^{-}$ but
$\hat Z$ is not $\epsilon$ located in either  $I_{z}^{+},\text{or} I_{z}^{-}$.The particle has a single trajectory in qr-number space that passes through the double slits and arrives at a single classical location on the detector screen. If the de Broglie relation $p_{y} = \frac{h}{\lambda_{y}}$ between the momentum $ p_{y}$ and a wave-length $\lambda_{y}$ is assumed then the difference between the qr-number lengths of the paths is approximately $\lambda_{y}$ when a particle arrives in the vicinity of the first maxima of the interference pattern. 

\subsubsection{ A deterministic  violation of Bell's theorem} Consider the Bell-Bohm experiment as a measurement of the first kind of the two particle attribute $ \hat C(\vec u_{L}, \vec u_{R}) = \sigma_{1}. \vec u_{L}\otimes \sigma_{2}. \vec u_{R}$ in which the measurement is repeated for many pairs, all prepared in the epistemic condition $W_{0}(\epsilon) = \nu (\rho_{0};\epsilon)$.
The qr-number value of the attribute $ \hat C(\vec u_{L}, \vec u_{R})$ for the $n^{th}$ pair in the condition $V_{n}(1,2)\subset W_{0}(\epsilon)$ is $c(\vec u_{L}, \vec u_{R})_{Q}(V_{n}(1,2))$. Now  $ \hat C(\vec u_{L}, \vec u_{R})^{2} = \hat I(1,2)$, that is  $ \hat C(\vec u_{L}, \vec u_{R})$ is a symmetry and $\|\hat C(\vec u_{L}, \vec u_{R})\| = 1$. Therefore  there exists a projection operator $\hat E(\vec u_{L}, \vec u_{R})$ such that $2\hat E(\vec u_{L}, \vec u_{R}) = \hat C(\vec u_{L}, \vec u_{R})+ \hat I(1,2)$.

With the ``ergodic assumption'', that $N$ independent measurements on one system are equivalent to one measurement on $N$ independent systems,  we can prove
\begin{theorem}
If the system is in the ontic  condition $V_{n}(1,2) = \nu (\rho_{n};\delta_{n}),\; \delta_{n}\ll \epsilon$,
then, up to $\delta_{n}$, the qr-number value of $\hat C(\vec u_{L}, \vec u_{R})$ in the condition $V_{n}(1,2)$ is approximate  equals
$Tr \rho_{n} \hat C(\vec u_{L}, \vec u_{R})$. That is, 
\begin{equation}|c(\vec u_{L}, \vec u_{R})_{Q}(V_{n})\ -\ Tr(\rho_{n} \hat C(\vec u_{L}, \vec u_{R}))|\
< \delta_{n}.\end{equation}  
\end{theorem}
\begin{proof}Let $\hat E(\vec u_{L}, \vec u_{R})$ be the projection operator of $\hat C(\vec u_{L}, \vec u_{R})$. Then the proof of Theorem 5 in \cite{durt2} can be used with $\hat E(\vec u_{L}, \vec u_{R})$ replacing the projection operator $\hat P_{i}(1) $.  
\end{proof}
Using the theory of independent errors, the most probable value of a set of  measurements is their arithmetical mean, as $\forall n\; : V_{n}(1,2)\subset W_{0}(\epsilon)$ and $|Tr \rho_{n} \hat C(\vec u_{L}, \vec u_{R}) -  Tr \rho_0 \hat C(\vec u_{L}, \vec u_{R})| < \epsilon$ the arithmetic mean will be approximately $Tr \rho_0 \hat C(\vec u_{L}, \vec u_{R})$. 
\begin{corollary} When $\rho_0 = \hat P_{\Psi_{\pm}^{s_1,s_2}}(1,2)$ then $Tr \rho_0 \hat C(\vec u_{L}, \vec u_{R}) = - \vec u_{L} \cdot \vec u_{R}$. Therefore in this qr-number deterministic model,  Bell's theorem is violated in accordance with experiment.\end{corollary}

Let $c_{n} = Tr \rho_{n} \hat C(\vec u_{L}, \vec u_{R}),\; n = 1,2,......N $ and $c_{0} =Tr \rho_{0} \hat C(\vec u_{L}, \vec u_{R}) $ then $|c_{n} - c_{0}| < \epsilon$ so if we put $c_{n} = c_{0} + \Delta_{n}$ then $|\Delta_{n}| < \epsilon$ for all $n$. Now the arithmetic mean of the $N$ measurements is $\sum_{\alpha = 1}^{N} \frac{c_{n}}{N} = c_{0} + \sum \frac{\Delta_{n}}{N} < c_{0} + \frac{\epsilon}{N}$ so if $\epsilon \ll 1$ and $N$ is large then the arithmetic mean is $c_{0}$ to a very good approximation.

\subsubsection{The covariance of qr-number values}
 
The position vector attribute of a particle with mass $m > 0$ and spin $s$ transforms covariantly under the Euclidean subgroup of its symmetry group. That is,  $\hat X_{i}^{\prime} = \mathcal{U}(R, \vec a)\hat X_{i} \mathcal{U}^{-1}(R, \vec a) =\sum_{j=1}^{3}R_{i j}\hat X_{j} + a_{i}$,  for $\vec a \in \RR^{3}$ and $R \in SO(3)$  this leads to a system of imprimitivity \cite{mack} for $\mathcal{U}(E)$. 

The covariance of a quality induces the covariance of its qr-number values,
\begin{equation}
\vec x_{Q}(W)^{\prime} = \vec x^{\prime}_{Q}(W) =  (R \vec x)_{Q}(W) + \vec a 1_{Q}(W) = \vec x_{Q}(W^{\prime})
\end{equation} if $1_{Q}(W)$ is the constant unit value qr-number on $W$ and $W^{\prime} = \mathcal{U}^{-1}(R, \vec a) W \mathcal{U}(R, \vec a)$.

At the microscopic level a quantum system can undergo a symmetry transformation in which the elements of  the group are given by qr-numbers, directly causing a transformation of the qr-number values of the qualities. Consider a qr-number realisation of the Euclidean group. For any $V \in \mathcal O(\EsubS(\mathcal{A})$ there is a spatial translation by $ \vec x_{Q}(V)$ that sends  $ \vec x_{Q}(W) \to  \vec x_{Q}(W) +  \vec x_{Q}(V)$ and for any matrix $R_{Q}(V) \in SO(3)( \RsubD(\EsubS(\mathcal A)))$, defined to extent $V$, that satisfies the conditions $R^{t}_{Q}(V) R_{Q}(V) = I_{Q}(V)\;, \text{det}  R_{Q}(V) = 1_{Q} (V)$, there is a rotation about the origin of the qr-number space given by $ \vec x_{Q}(W) \to  R_{Q}(V)\vec x_{Q}(W)$ for any $W \in O(\EsubS)$.

Therefore under the qr-number Euclidean transformation $(\vec x_{Q}(U), R_{Q}(V))$, the vector $\vec x_{Q}(W)$ transforms to  
\begin{equation}
\vec x_{Q}^{\prime}(W^{\prime}) =   R_{Q}(V))\vec x_{Q}(W) + \vec x_{Q}(U)
\end{equation}
where $W^{\prime} = (V \cap W)\cup U$.
Since the product of $R_{Q}(V)$ with $\vec x_{Q}(W)$ is only non-zero if $W \cap V\neq \emptyset $ the $R_{Q}(V)$ only rotates particles located at positions $\vec x_{Q}(U)$ with $U\cap V \neq \emptyset$.

\subsubsection{Heisenberg Inequalities for qr-numbers}This gives a limitation on the accuracy with which the qr-number values of two attributes represented by non-commuting operators can be simultaneously approximated by standard real number values on an open set $W$.  
Let $\hat A,\; \hat B \in dU(\mathcal{E}(\mathcal{G}))$ be essentially self adjoint on  $\mathcal{D}^{\infty}(U) \subset \mathcal{H}$.
\begin{theorem} If both $\hat A$ is $\epsilon$-sharp collimated in $I_{a}$ and $\hat B$ is $\epsilon$-sharp collimated in $I_{b}$ when the system is the condition $W$ and  $\imath \hat C = [\hat A, \;\hat B] $ then  \begin{equation} |I_{a}||I_{b}| \geq 2|c_{Q}(W)|/\epsilon \end{equation}\end{theorem}

Since the width of a slit gives
a measure of the accuracy of the approximate standard real number values when the quantity is $\epsilon$-sharp collimated in it, the Heisenberg's uncertainty principle limits in accuracy at which the attributes can be realised simultaneously in the condition $W$.

\begin{corollary}\cite{durt2}
Let $\hat Q$ and $\hat P$ represent the position and its conjugate momentum of a
massive particle, i.e., $\imath [ \hat P, \; \hat Q] = \hbar$, and let $I_{q}$ and $I_{p}$
be slits for the conjugate variables. If a particle in a condition $W$ is $\epsilon$-sharp
collimated through both slits then the product of the widths of the
slits must satisfy,\begin{equation}| I_{q}||I_{p}| \geq \ 2\hbar/ {\epsilon }
\end{equation}
\end{corollary}

This result determines the minimum area in the classical phase space
that is required if a particle is to be $\epsilon$-sharp collimated
in both the $\hat Q$ and $\hat P$ attributes. This
inequality does not explicitly restrict the qr-numbers
$q_Q(W)$ and $p_Q(W)$ but if $W_{q}$ and $W_{p}$ are the largest open sets on which $\hat Q$ and $\hat P$ are respectively $\epsilon$-sharp collimated in $I_{q}$ and $I_{p}$, then  $W = W_{q} \cap W_{p}$ must be non-empty.  

\subsubsection{The L\"uders-von Neumann transformation rule.\label{ludo}}

In the standard theory, the collapse hypothesis is presented as an
independent postulate governing the behaviour of systems undergoing
measurement. It states that if the system was prepared in the quantum
state $ \rho_0$, then in a measurement of $\hat A$, instead of evolving unitarily,
$ \rho_0$ ``collapses'' to
$ \rho_{0}^{\prime} = { \hat P^{A}(I_{a}) \rho_0 \hat P^{A}(I_{a})\over Tr( \hat P^{A}(I_{a})\rho_0)}$, where $ \hat P^{A}(I_{a})$ is the projection operator of $\hat A$ onto the interval $I_{a}$ and $\hat P^{A}(I_{a})\rho_0 \ne  0$. In this process any $\hat B \in \mathcal{A}$ changes
to ${ \hat P^{A}(I_{a}) \hat B \hat P^{A}(I_{a})\over Tr( \hat P^{A}(I_{a})  \rho_0)}$.

In \cite{durt2} we proved that if the qr-number value $a_Q(W)$ of $\hat A$ is strictly $\epsilon$ sharply collimated in an interval $I_{a}$ then the qr-number value of any attribute changes as if it
had undergone a L\"uders-von Neumann transformation. This does not mean that
the collapse process really occurs, but rather that the collapse
postulate gives a good approximation to any attribute's qr-number value that is obtained in this type of measurement.
\begin{theorem} L\"uders-von Neumann rule for strictly $\epsilon$
sharp preparations \cite{durt2}.
If initially the system is prepared in $W=\nu(\rho_0, \delta)$ and subsequently $\hat A$ is strictly
$\epsilon$ sharply collimated in $I_{a}$ on $U$, then any $\hat B\in \mathcal{A}$ will have the qr-number value $b_Q(U \cap W) \approx (Tr  \rho_{0}^{\prime} \hat B)1_Q(U \cap W)$ with an accuracy proportional to $\delta + 2 \epsilon$.
\end{theorem} The accuracy is controlled by the choice of $\delta$ and $\epsilon$. If the ontic condition of a particle is $V \subset W$ then the attribute $\hat B$ will have the qr-number value  $b_Q(U \cap V)\approx (Tr  \rho_{0}^{\prime} \hat B)1_Q(U \cap V)$.

\subsubsection{Equations of motion for qr-numbers.}

For any open set $U$, the qr-number values $ (\vec q_{Q}(U)(t),\vec p_{Q}(U)(t)) $  of the position and momentum of a massive quantum particle satisfy classical equations of motion.
Thus, if $ h(\vec q_Q(U)(t),\vec p_Q(U)(t))$ is the qr-number value of the Hamiltonian,
\begin{equation}\label {EQ12}
{dq^{j}_{Q}(U)(t)\over dt} =  
{\partial h(\vec q_{Q}(U)(t),\vec p_{Q}(U)(t))
\over \partial p^{j}_{Q}(U)(t)},
\end{equation}
\begin{equation}\label{EQ13}
{dp^{j}_{Q}(U)(t)\over dt} = - {\partial h(\vec q_{Q}(U)(t),\vec p_{Q}(U)(t))
\over \partial q^{j}_{Q}(U)(t)}.
\end{equation}
$h(\vec q_Q(U)(t),\vec p_Q(U)(t)) = \sum_{j=1}^{3} {1\over 2m}(p^{j}_Q(U)(t))^{2} + V(\vec q_Q(U)(t)) $ and ${d\over dt}$ denotes differentiation with respect to time.

Distinguishing a quality from its qr-number value, when a particle is in a condition $W$, the equations of motion for $\vec q_{Q}(W)(t)$ are Hamiltonian while those for $\vec q(t)_{Q}(W)$ come from averaging Heisenberg operator equations over $W$. There are two possible qr-number equations of motion for the particle, even though for the standard Hamiltonian $\hat H$ the  ${d q^{j}_{Q}(W)\over dt}(t) = ({ dq^{j}(t)\over dt})_{Q}(W) = {1\over m_{j}}p^{j}(t)_{Q}(W)$.  

The different evolutions give the same trajectory, $\vec q_{Q}(W)(t) = \vec q(t)_{Q}(W)$, when the interaction term in the Hamiltonian is linear: examples are the zero force law (free motion) and linear force laws, such as the simple harmonic motion.   Furthermore, for suitably smooth forces, these evolutions are locally indistinguishable in the sense that there exists a class of open subsets $W(\vec x, \epsilon)$ of state space $\EsubS(\mathcal{A})$ on which the averaged values of 
Heisenberg's operator equations closely approximate Hamilton's equations for the qr-number values.  
  
The class $W(\vec x, \epsilon)$ for $\vec x \in \RR^{3}$ and $\epsilon > 0$ do not cover $\EsubS(\mathcal{A})$, but  associated to each
$W(\vec{x},\epsilon)$ is an open ball
$B(\vec{x},\delta)$ in $\bf R^{3}$, the collection of which cover
$\RR^{3}$. An observer measuring a particle
with apparatus set up in one of these open balls
could not determine locally whether the evolution
of the particle was governed by Heisenberg's operator equations
of motion modulated by states $\rho$ from
$W(\vec{x},\epsilon)$ or by Hamilton's equations of motion
for the qr-numbers $\vec{Q}\mid_{W}$ restricted to
the open set $ W = W(\vec{x},\epsilon)$.
\subsection{Conclusions}
The view presented in this paper is that the change in going from a mathematical description of the macroscopic to the microscopic world requires a change in the system of real numbers that are taken as numerical values by the physical qualities in the world. If the Dedekind real numbers hypothesis is correct then there should be a postulate of covariance under change of topos in which the sheaf of Dedekind real number exists because maps between $\RsubD(X)$ and $\RsubD(Y)$ are obtained using functors between the toposes $Shv(X)$ and $Shv(Y)$.
This postulate would be along the lines that the general laws of physics should be expressible in equations which hold good for all systems of Dedekind real numbers. We believe that at least the equations of motion for both quantum and classical particles can be expressed in Dedekind real numbers in both the Lagrangian and Hamiltonian formulations. 

\subsection{Acknowledgments}
The paper is based on much collaborative work \cite{durt2, JVC1,JVC2}  with
Thomas Durt  whose important contributions are acknowledged. I also wish to
thank Ross Street and members of the Centre of Australian Category Theory for
their encouragement and Huw Price and members of the Foundations of QFT group
for interesting discussions. 

%\bibliography{your bib file}

\end{document}